\documentclass[runningheads]{llncs}

\usepackage{graphicx}
\usepackage{amsmath,amssymb}
\usepackage{url}
\usepackage{tikz}

\newcommand{\EE}{\mathbb{E}} 
\newcommand{\RR}{\mathbb{R}}
\renewcommand{\SS}{\mathbb{S}}

\def\cqedsymbol{\ifmmode$\lrcorner$\else{\unskip\nobreak\hfil
\penalty50\hskip1em\null\nobreak\hfil$\lrcorner$mail
\parfillskip=0pt\finalhyphendemerits=0\endgraf}\fi}

\DeclareMathOperator{\cro}{cr}

\usepackage{tikz-3dplot}
\usetikzlibrary{intersections}

\newcommand{\InterSec}[3]{%
    \path[name intersections={of=#1 and #2, by=#3, sort by=#1,total=\t}]
        \pgfextra{\xdef\InterNb{\t}}; }

\begin{document}


\title{Limiting crossing numbers for geodesic drawings on the sphere}

\author{Marthe Bonamy\thanks{Supported in part by the ANR Project DISTANCIA (ANR-17-CE40-0015) operated by the French National Research Agency (ANR).}
\inst{1}
\and
  Bojan Mohar\thanks{Supported in part by the NSERC Discovery Grant R611450 (Canada), by the Canada Research Chairs program, and by the Research Project J1-8130 of ARRS (Slovenia).}
\inst{2}
\and
  Alexandra Wesolek\thanks{Supported by the Vanier Canada Graduate Scholarships program.}
\inst{2}
}

\institute{CNRS, LaBRI, Universit\'e de Bordeaux, France \\
  \email{marthe.bonamy@u-bordeaux.fr}
  \and
  Department of Mathematics, Simon Fraser University, Burnaby, BC, Canada\\
  \email{mohar@sfu.ca}\\
  \email{agwesole@sfu.ca}
  }
\maketitle

\begin{abstract}
We introduce a model for random geodesic drawings of the complete bipartite graph $K_{n,n}$ on the unit sphere $\SS^2$ in $\RR^3$, where we select the vertices in each bipartite class of $K_{n,n}$ with respect to two non-degenerate probability measures on $\SS^2$. It has been proved recently that many such measures give drawings whose crossing number approximates the Zarankiewicz number (the conjectured crossing number of $K_{n,n}$). In this paper we consider the intersection graphs associated with such random drawings. We prove that for any probability measures, the resulting random intersection graphs form a convergent graph sequence in the sense of graph limits. The edge density of the limiting graphon turns out to be independent of the two measures as long as they are antipodally symmetric. 
However, it is shown that the triangle densities behave differently. We examine a specific random model, blow-ups of antipodal drawings $D$ of $K_{4,4}$, and show that the triangle density in the corresponding crossing graphon depends on the angles between the great circles containing the edges in $D$ and can attain any value in the interval $\bigl(\frac{83}{12288}, \frac{128}{12288}\bigr)$.
\end{abstract}

\section{Introduction}

The crossing number $cr(G)$ of a graph $G$ is the minimum number of crossings obtained by drawing $G$ in the plane (or the sphere). In this paper we consider the \emph{(spherical) geodesic crossing number} $cr_0(G)$, for which we minimize the number of crossings taken over all drawings of $G$ in the unit sphere $\SS^2$ in $\RR^3$ such that each edge $uv$ is a geodesic segment joining points $u$ and $v$ in $\SS^2$. Recall that \emph{geodesic segments} (or \emph{geodesic arcs}) in $\SS^2$ are arcs of great circles whose length is at most $\pi$. 
Also note that $cr(G) \le cr_0(G)$ for every graph $G$.

Crossing number minimization has a long history and is used both in applications and as a theoretical tool in mathematics. We refer to \cite{SchBook18} for an overview about the history and the use of crossing numbers. Despite various breakthrough results about crossing numbers, some of the very basic questions remain open as of today, two of the most intriguing being what are the crossing numbers of the complete graphs $K_n$ and what are the crossing numbers of the complete bipartite graphs $K_{n,n}$ (the Tur\'an Brickyard Problem). The asymptotic versions of both problems are strongly related \cite{richter1997relations} and a lower bound for the limiting crossing number of $K_{n,n}$ gives a related lower bound for $K_n$. The asymptotic version of the rectilinear crossing number of $K_n$ is related to Sylvester's Four point problem in the plane \cite{Sylvester65,SchWi94}, see also \cite{SchBook18} for recent results. The geodesic version on the sphere, which we discuss in this paper, is a spherical version of Sylvester's problem.

\subsection{Outline}

In this paper we initiate the study of limiting properties of intersection graphs associated with drawings of complete and complete bipartite graphs. We limit ourselves to geodesic drawings on the unit sphere in $\RR^3$ in which case the drawings are determined by the choice of the placements of the vertices on the sphere. The first main result of this work shows that whenever the vertices in each bipartite class of $K_{n,n}$ are selected according to some (non-degenerate) probability measure on $\SS^2$ (where the two measures used for each class can be different), then, with probability 1, the intersection graphs form a convergent sequence of graphs in the sense of graph limits \cite{LovaszBook}. See Theorem \ref{thm:convergence}.

The basic combinatorial property of convergent graph sequences is that of subgraph densities. The density of edges in the crossing graphs corresponds to the asymptotic crossing number. In addition to this, we examine one particular related basic question: what is the density of triangles. 
We show that their density can be substantially different among different randomized models. 
Although this result may be seen as ``expected'', it is still somewhat surprising. Indeed, it shows that there is a large variety of drawings of $K_{n,n}$, all attaining the Zarankiewicz bound, 
in which the number of triples of mutually crossing edges varies significantly, and can attain any value in the interval $\bigl(\frac{83}{12288}, \frac{128}{12288}\bigr)$.
See Theorems \ref{thm:K3 blowup} and \ref{thm:trianglebounds}. 
We believe that further exploring of subgraph densities in crossing graphons may give a deeper insight into the basic Tur\'an's Brickyard Problem for geodesic drawings on the sphere.

\subsection{Asymptotic Zarankiewicz Conjecture}

During World War \MakeUppercase{\romannumeral 2}, Hungarian mathematician P\'{a}l Tur\'{a}n worked in a brick factory near Budapest. There the bricks were transported on wheeled trucks from kilns to storage yards. It was difficult to push the trucks past the rail crossings and it would result in extra work if bricks fell of the trucks. Therefore Tur\'{a}n wondered if there was a way of arranging the rails such that there would be less crossings between them. Seeing the kilns and storage yards as parts of a bipartite graph, this led to the more general question of the minimum number of crossings in drawings of complete bipartite graphs $K_{n,n}$. Zarankiewicz \cite{zarankiewicz1955problem} and Urbanik \cite{urbanik1955solution} suggested drawings that involved
\begin{equation}\label{eq:Hill}
   Z(m,n) = \lfloor \tfrac{n}{2}\rfloor \, \lfloor \tfrac{n-1}{2}\rfloor \, \lfloor \tfrac{m}{2}\rfloor \,\lfloor \tfrac{m-1}{2}\rfloor =
   \left\{
       \begin{array}{ll}
         \tfrac{1}{16} n(n-2)m(m-2), & \hbox{$n,m$ are even;} \\[1mm]
         \tfrac{1}{16} n(n-2)(m-1)^2, & \hbox{$n$ is even, $m$ is odd;}\\[1mm]
         \tfrac{1}{16} (n-1)^2(m-1)^2, & \hbox{$n,m$ are odd}
       \end{array}
   \right.
\end{equation}
crossings. Whether this value is the best possible remains unanswered to this day despite numerous attacks using powerful machinery in trying to resolve this conjecture. 


A general construction of drawings of complete bipartite graphs attaining the Zarankiewicz bound was recently exhibited~\cite{Mo19}. All of them are geodesic drawings in $\SS^2$ and they show that
\begin{equation}
   cr(K_{n,n})\le cr_0(K_{n,n}) \le Z(n,n) \quad \hbox{for every } n\ge 1.
   \label{eq:crbounds}
\end{equation}

It is not hard to see that the following limits exist: 
$$
   \lambda := \lim_{n\to\infty}  n^{-4}\, cr(K_{n,n}) \quad \hbox{and} \quad \lambda_0 := \lim_{n\to\infty} n^{-4}\, cr_0(K_{n,n}).
$$
Clearly, (\ref{eq:crbounds}) implies that $\lambda \le \lambda_0 \le \tfrac{1}{16}$. The \emph{asymptotic Zarankiewicz conjecture} for the usual and the geodesic crossing number is also open.

\begin{conjecture}
\label{conj:asymptotic Zarankiewicz conjecture}
$\lambda = \lambda_0 = \tfrac{1}{16}$.
\end{conjecture}

\subsection{Random drawings of complete bipartite graphs}

In 1965, Moon \cite{Moon65} proved that a random set of $n$ points on the unit sphere $\SS^2$ in $\RR^3$ joined by geodesics gives rise to a drawing of $K_n$ whose number of crossings asymptotically approaches the conjectured value. It was proved recently \cite{Mo20} that the same phenomenon appears in a much more general random setting. These results can also be extended to random drawings of the complete bipartite graphs $K_{n,n}$ where it was shown that under a symmetry condition on the probability measures the crossings in such drawings converge to the Zarankiewicz value.

A probability distribution $\mu$ on $\SS^2$ is \emph{nondegenerate} if for every great circle $Q\subset \SS^2$, $\mu(Q)=0$. It is \emph{antipodally-symmetric} if for every measurable set $A\subseteq \SS^2$ the measure of its antipodal set $\overline{A}$ is the same, $\mu(A)=\mu(\overline{A})$. 

\begin{theorem}[\cite{Mo20}]
\label{thm:Mo}
Let $\mu_1,\mu_2$ be nondegenerate antipodally-symmetric probability distributions on the unit sphere $\SS^2$. Then a $\mu_1$-random set of $n$ points on $\SS^2$ joined by geodesics (segments of great circles) to a $\mu_2$-random set of $n$ points gives rise to a drawing $D_n$ of the complete bipartite graph $K_{n,n}$ such that $cr(D_n)/Z(n,n)=1 + o(1)$ a.a.s.
\end{theorem}

The random drawing model in the theorem will be referred to as \emph{$(\mu_1,\mu_2)$-random drawing} of the complete bipartite graph $K_{n,n}$.

\subsection{Crossing graphon}

Let $N=\{n_1,n_2,n_3,\dots\}$ be an infinite set of positive integers, where $n_1 < n_2 < n_3 < \cdots$. Suppose that for each $n\in N$, we have a drawing $D_n$ of $K_{n,n}$. To each such drawing we associate the \emph{crossing graph} $X_n=X_n(D_n)$, whose vertices are all $n^2$ edges in $D_n$, and two of them are adjacent in $X_n$ if they cross in $D_n$. Then we can consider what may be the \emph{limit} of the sequence $(X_n)_{n\in N}$. The notion of \emph{graph limits} has been introduced by Lov\'asz et al. \cite{BCLSV06,BCLSV08,LSz06}, see \cite{LovaszBook}. The basic setup is described below.

Let $(X_n)_{n\in N}$ be a sequence of graphs. For any fixed graph $H$, let $k=|H|$ be its order, and let $hom(H,X_n)$ denote the number of graph homomorphisms $H\to X_n$, i.e. the number of maps $\phi: V(H)\to V(X_n)$ such that for each edge $uv\in E(H)$, $\phi(u)\phi(v) \in E(X_n)$. Then we define the \emph{homomorphism density} for $H$ as
$$
    t(H,X_n) = \frac{hom(H,X_n)}{|X_n|^k}.
$$
Note that this is the probability that a random mapping $V(H)\to V(X_n)$ is a homomorphism. If the sequence $t(H,X_n)$ converges, we denote its limit by $t(H)$. If $t(H)$ exists for every $H$, then we say that $(X_n)$ is a \emph{convergent sequence of graphs}. In that case there is a well-defined object $W$, called a \emph{graphon}, and the graphon $W$ is called the \emph{limit} of this convergent sequence \cite{BCLSV06,BCLSV08}. We define the \emph{homomorphism densities} of $W$ by setting $t(H,W) = \lim_{n\to \infty} t(H,X_n) = t(H)$.

The space of all graphons is a compact metric space \cite{LovaszBook,LSz06}.
Given any graphon $W$, one can define \emph{$W$-random graphs} \cite{LS08}. A sequence $(R_n)$ of $W$-random graphs is convergent with probability 1, and its limit is $W$.

In this paper we consider nondegenerate probability measures on $\SS^2$. For each pair of such probability measures $\mu_1$ and $\mu_2$, we have a $(\mu_1,\mu_2)$-random sequence of drawings $D_n$ of complete bipartite graphs $K_{n,n}$ and we consider their crossing graphs $X_n$. We prove that these sequences are convergent with probability 1 and discuss their homomorphism densities with the goal to better understand Conjecture \ref{conj:asymptotic Zarankiewicz conjecture}.

\begin{theorem}
\label{thm:convergence}
Let $\mu_1$ and $\mu_2$ be nondegenerate probability measures on $\SS^2$. Let $A_n$ and $B_n$ be a $\mu_1$-random and a $\mu_2$-random set of $n$ points in $\SS^2$, respectively, let $D_n$ be the corresponding $(\mu_1,\mu_2)$-random geodesic drawing of $K_{n,n}$ on parts $A_n$ and $B_n$, and let $X_n$ be its crossing graph. The sequence of graphs $(X_n)$ is convergent with probability $1$ and there is a graphon $W=W(\mu_1,\mu_2)$ that is the limit of this convergent sequence.
\end{theorem}

Since the number of edges in the crossing graph corresponds to the number of crossings in $D_n$, we have 
$$
   t(K_2, X_n) = \frac{2|E(X_n)|}{|X_n|^2}  = \frac{2cr(D_n)}{n^4}.
$$
Thus, Theorem \ref{thm:Mo} shows a tight relationship with the asymptotic Zarankiewicz conjecture and can be expressed as follows.

\begin{theorem}
\label{thm:K2 density}
Let $\mu_1, \mu_2$ be nondegenerate antipodally-symmetric probability measures on $\SS^2$. Let $W=W(\mu_1,\mu_2)$ be the corresponding graphon of the sequence $(X_n)$ as defined above. Then $$t(K_2,W(\mu_1,\mu_2))=\frac{1}{8}.$$
\end{theorem}

\subsection{Definitions}

We follow standard terminology from \cite{BM08,Di05} for graph theory and from \cite{SchBook18} for drawings of graphs. A drawing of a graph is \emph{good} if any two edges cross at most once, no two edges with a common endvertex cross, and no three edges cross at the same point. The first two conditions are clear when we consider geodesic drawings, and the third condition can always be satisfied if we make an infinitesimal perturbation. 

We say that a set of points on the unit sphere $\SS^2$ is \emph{in general position} if no two of the points are antipodal to each other, no three of them lie on the same great circle and no three geodesic arcs joining pairs of points cross at the same point. If $\mu$ is a nondegenerate probability distribution on $\SS^2$, then randomly chosen vertices will be in general position with probability 1.


\section{The proof of Theorem \ref{thm:convergence}}
\label{sec:proof Th2}

In the following we want to draw a comparison of subgraph densities of the crossing graphs $X_n$ to a concept similar to the Buffon Needle Problem (see, e.g. \cite{Isokawa99} or \cite{WeTr94}). We pick endpoints of segments randomly w.r.t. some probability distribution and consider the crossings formed by the segments. If the probability distribution is uniform on the sphere, it is equivalent as throwing a (bended) needle onto the sphere, where the needle length varies. Now considering a small number of such segments on the sphere we ask how they will cross each other.

Let $\mu_1$ and $\mu_2$ be nondegenerate probability measures on $\SS^2$. A \emph{$(\mu_1, \mu_2)$-random geodesic segment} is a geodesic segment $uv$ whose endpoints $u,v$ are chosen randomly w.r.t. $\mu_1$ and $\mu_2$, respectively. For a given graph $H$ of order $k=|H|$, we pick $k$ $(\mu_1, \mu_2)$-random geodesic segments on the sphere and look at the probability that $H$ is a subgraph of their intersection graph. Let $A=\{a_1, \dots, a_k\}$ be a $\mu_1$-random set of points in $\SS^2$ and $B=\{b_1, \dots,b_k\}$ be a $\mu_2$ random set of points in $\SS^2$. The segments we are considering are $a_1b_1,\dots,a_k b_k$. Note that the probability that $H$ is a subgraph of the intersection graph of $a_1b_1,\dots,a_k b_k$ depends on $\mu_1$ and $\mu_2$ only.

\begin{definition}
\label{def:pH}
Let $X$ be the intersection graph of $k$ $(\mu_1, \mu_2)$-random geodesic segments $a_1b_1,\dots,a_k b_k$ and let $H$ be a graph of order $k$. For a bijection $\phi:V(H) \to V(X)$ we define 
$$
p_H:=Pr[ \phi \text{ is a graph homomorphism}].
$$
\end{definition}
Observe that $p_H$ is independent of $\phi$, since the segments $a_ib_i$ ($i=1,\dots,k$) are selected independently. 

We want to compare the above model with another model where we pick $n \gg k$ points with respect to $\mu_1$ and $\mu_2$ each, and consider the corresponding crossing graph $X_n$ of a drawing $D_n$ of $K_{n,n}$. We will show that the models are closely related:  with growing $n$, picking $k$ vertices from $X_n$, they will with high probability come from $k$ independent geodesic segments and therefore represent $(\mu_1, \mu_2)$-random geodesic segments.  In the following we fix a graph $H$ and a mapping $\phi:V(H) \to V(X_n)$.

\begin{definition}
\label{def:y}
For given $X_n$, let $\phi: V(H)\to V(X_n)$ and we define the random variable $y_{H,\phi}$ on $X_n$ to be
\begin{align*}
    y_{H,\phi}(X_n) = \begin{cases}
       1 & \text{ if } \phi \text{ is a graph homomorphism }H\to X_n \\
       0 & \text{otherwise}
    \end{cases}
\end{align*}
and denote its expectation by
\begin{align*}
  E_\phi:= \EE[y_{H,\phi}].
\end{align*}
\end{definition}

Note that $E_\phi$ is not the same for every $\phi$. For example, if $H$ is a complete graph, then $E_\phi=0$ whenever $im(\phi)$ contains edges that share a vertex, as those edges never cross and hence are not adjacent in the crossing graph. 


\begin{lemma}
\label{lem:pH}
Let $(X_n)$ be a sequence of the crossing graphs of $(\mu_1,\mu_2)$-random geodesic drawings $D_n$ of $K_{n,n}$ for $n=1,2,\dots$, and let $H$ be a fixed graph of order $k$. Then
\begin{align*}
\lim_{n \to \infty} \ \frac{1}{\vert X_n\vert^k} \sum_{ \phi: V(H)\to V(X_n)} E_\phi = p_H.
\end{align*}
\end{lemma}

\begin{proof}
Let $im(\phi)=\{v_1w_1,\dots,v_kw_k\}$. Then if $\left\vert \{v_1,\dots,v_{k},w_1,\dots,w_{k}\} \right\vert=2n$ we are in the setup of Definition $\ref{def:pH}$ and $\EE[y_{H,\phi}]=p_H$. Moreover, there are $O(n^{2k-1})$ choices for $\phi$ for which $\left\vert \{v_1,\dots,v_{k},\allowbreak w_1,\dots,w_{k}\} \right\vert<2n$  and the result follows.\qed
\end{proof}

Let us now consider the sum of the above defined random variables
\begin{equation}\label{eq:Y_H}
    Y_H:=\sum_{ \phi: V(H)\to V(X_n)} y_{H,\phi},
\end{equation}
and note that $Y_H(X_n)=hom(H,X_n)$ and $\EE[Y_H]=\sum_{ \phi: V(H)\to V(X_n)} E_\phi$. The aim is to show that $Y_H$ is in general not far from its expectation. This then gives us the tool to show the existence of $ \lim_{n \to \infty} \frac{\vert Y_H \vert}{\vert X_n \vert^k} = t(H)$ with probability 1.

\begin{proposition}\label{prop:varY_H}
Let\/ $Y_H$ be defined as in (\ref{eq:Y_H}). Then we have 
\begin{align*}
var(Y_H) =O(n^{4 k-2}).
\end{align*}
\end{proposition}

The proof of the proposition is in the appendix.

\begin{proof}[of Theorem \ref{thm:convergence}]
By Proposition \ref{prop:varY_H} and Chebyshev's inequality there exists a constant $C$ such that
\begin{align*}
    Pr\left[\vert Y_H - \EE[Y_H] \vert \geq k C n^{2k-1}  \right] \leq \frac{1}{k^2}.
\end{align*}
Now if we choose $k=k(n)$ appropriately such that $k(n) n^{-1}$ converges to zero and the sum $\sum_{n=1}^\infty \frac{1}{k(n)^2}$ is finite we can use the Borel-Cantelli Lemma. For example, we can choose $k=n^{3/4}$ and using Lemma $\ref{lem:pH}$ we get
\begin{align*}
    Pr\left[\left\vert\frac{\vert Y_H \vert}{\vert X_n \vert^{k}}-p_H\right\vert -\left\vert p_H-\frac{\EE[Y_H]}{\vert X_n \vert^{k}}\right\vert \geq \frac{  C n^{2k-1/4}}{\vert X_n \vert^{k}}  \right] &\leq \frac{1}{n^{3/2}}\\
   \implies Pr\left[\left\vert\frac{\vert Y_H \vert}{\vert X_n \vert^{k}}-p_H\right\vert \geq \frac{C'}{n^{1/4}}\right] &\leq \frac{1}{n^{3/2}}.\\
\end{align*}
for some constant $C'$. Then the Borel-Cantelli Lemma implies the following.

\begin{claim}
For each fixed $H$, $\frac{\vert Y_H \vert}{\vert X_n \vert^{k}} \to p_H:=t(H)$ with probability $1$.
\end{claim}
Given that for each $H$, $t(H, X_n) \to t(H)$ with probability $1$, and since the probabilities are countably additive, it follows with probability $1$ that $t(H, X_n) \to t(H)$ for every $H$. Consequently, the sequence of random crossing graphs $(X_n)$ is convergent with probability $1$.\qed
\end{proof}

\section{Blowup of an antipodal drawing of $K_{4,4}$}

In the previous sections, we have established the existence of crossing graphons and determined densities $t(H)$ for $H=K_2$ if our measures $\mu_1,\mu_2$ are antipodally symmetric. Somewhat surprisingly, these edge densities are the same for any ``suitable'' measures $\mu_1,\mu_2$. It is natural to ask what happens with other homomorphism densities in these crossing graphons. 
The purpose of this section is to show that the homomorphism densities of triangles behave differently. To us, this was not {\it a priori} clear. We study a particular case of $(\mu_1,\mu_2)$-random drawings of complete bipartite graphs and determine $t(K_3)$ for the corresponding graphon $W(\mu_1,\mu_2)$.

In the following we fix a drawing $D_4$ of the complete bipartite graph $K_{4,4}$ where each part consists of two antipodal pairs of vertices on $\SS^2$ as in Figure $\ref{fig:sphere}$.

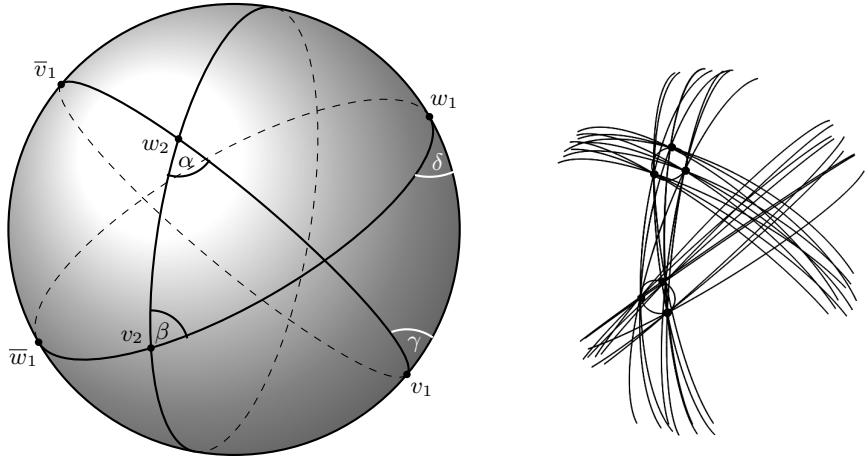
\begin{figure}
    \centering
\begin{tikzpicture}
   \pgfmathsetmacro\R{3} 
    \hspace*{-3cm}
\shadedraw[ball color = white,thick] (0,0) circle (\R);

        \foreach \angle[count=\n from 1] in {30,260} {

          \begin{scope}[rotate=\angle]
            \path[draw,dashed,name path global=d\n] (\R,0) arc [start angle=0,
              end angle=180,
              x radius=\R cm,
            y radius=1cm] ;
            \path[draw,thick,name path global=s\n] (-\R,0) arc [start angle=180,
              end angle=360,
              x radius=\R cm,
            y radius=1cm] ;
          \end{scope}
        }

          \begin{scope}[rotate=-40]
            \path[draw,thick,name path global=d3] (\R,0) arc [start angle=0,
              end angle=180,
              x radius=\R cm,
            y radius=0.5cm] ;
            \path[draw,dashed,name path global=s3] (-\R,0) arc [start angle=180,
              end angle=360,
              x radius= \R cm,
            y radius=0.5cm] ;
          \end{scope}

    \InterSec{s1}{s2}{I3} ;
    \InterSec{s1}{d3}{I2} ;
    \InterSec{d3}{s2}{I1} ;
    
    \InterSec{d1}{d2}{J3} ;
    \InterSec{d1}{d3}{J2} ;
    \InterSec{d3}{d2}{J1} ;    
    \fill[black] (I3) circle (0.05 cm) node[xshift=-0.25cm, yshift=0.15cm]{$v_2$};  
  \fill[black] (I1) circle (0.05 cm)node[xshift=-0.3cm, yshift=-0.1cm]{$w_2$};  ;  
      \fill[black] (3*cos{140}, 3*sin{140}) circle (0.05 cm) node[xshift=-0.2cm, yshift=0.2cm]{$\overline{v}_1$};  
            \fill[black] (3*cos{320}, 3*sin{320}) circle (0.05 cm) node[xshift=0.2cm, yshift=-0.2cm]{$v_1$}; 
                  \fill[black] (3*cos{210}, 3*sin{210}) circle (0.05 cm) node[xshift=-0.2cm, yshift=-0.2cm]{$\overline{w}_1$};  
            \fill[black] (3*cos{30}, 3*sin{30}) circle (0.05 cm) node[xshift=0.2cm, yshift=0.2cm]{$w_1$}; 
            \draw[thick, shift=(I1)] plot [domain=252:326] ({0.5*cos(\x)}, {0.5*sin(\x)});
            \node[xshift=0.1cm,yshift=-0.3cm] at (I1) {$\alpha$};
            \draw[thick, shift=(I3)]  plot [domain=15:90] ({0.5*cos(\x)}, {0.5*sin(\x)});
            \node[xshift=0.15cm,yshift=0.2cm] at (I3) {$\beta$};
            \draw[thick,white] plot [domain=55:110] ({0.6*cos(\x)+3*cos(320)}, {0.6*sin(\x)+3*sin(320)});
            \node[xshift=0.1cm,yshift=0.4cm,white] at (3*cos{320}, 3*sin{320}) {$\gamma$};
            \draw[thick,white]  plot [domain=256:360] ({0.8*cos(\x)+3*cos(30)}, {0.8*sin(\x)+3*sin(30)});
            \node[xshift=0.11cm,yshift=-0.6cm,white] at (3*cos{30}, 3*sin{30}) {$\delta$};

           \hspace*{6.4cm}
            
\node[inner sep=0pt] at (0,-0.3)
    {\includegraphics[width=.35\textwidth]{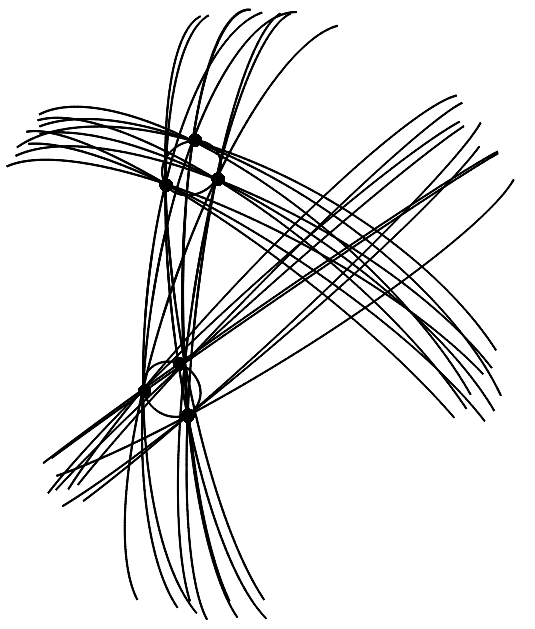}};
    
\end{tikzpicture}
\caption{The left part shows a drawing $D_4$ of a $K_{4,4}$ on parts $\{v_1,\overline{v}_1, v_2, \overline{v}_2\}$ and $\{w_1,\overline{w}_1, w_2, \overline{w}_2\}$. The angles $\alpha$ and $\beta$ are in the triangle formed by $w_2$, $v_2$ and a crossing, whereas $\gamma$ and $\delta$ are in a triangle formed by $v_1,w_1$ and the same crossing. The right-hand side shows part of a $D_4^{(3)}$ drawing with the circles of $w_2$ and $v_2$ each containing 3 vertices and with nine edges for each incident bundle emanating from these two nodes.}
\label{fig:sphere}
\end{figure}
    
We will be considering a \emph{blowup drawing} $D_4^{(n)}$ of $D_4$ for which we replace each vertex from $D_4$ with a circle of some small radius $r=r(n)$ that is centered at that vertex, and position $n$ evenly spaced vertices on that circle. These $n$ vertices will be referred to as the \emph{node} of the corresponding vertex of $K_{4,4}$. We also assume that all $8n$ vertices obtained in this way are in general position. In that way, each edge of $K_{4,4}$ is replaced by a complete bipartite graph between the corresponding nodes which we call the \emph{edge bundle}. This means for $N=4n$ that $D_4^{(n)}$ is a drawing of $K_{N,N}$. In what follows, we discuss the number of triangles in the intersection graph (of edges in $D_4^{(n)}$) when $n$ grows large. To simplify our discussion about triangles, we first classify the crossings in $D_4^{(n)}$.

\subsection{Types of crossings in $D_{4}^{(n)}$}\label{subsec:crkkn}

In the blowup drawing $D_4^{(n)}$, we distinguish three types of crossings, depending on what they stem from, as depicted in Figure \ref{fig:crossings}.

   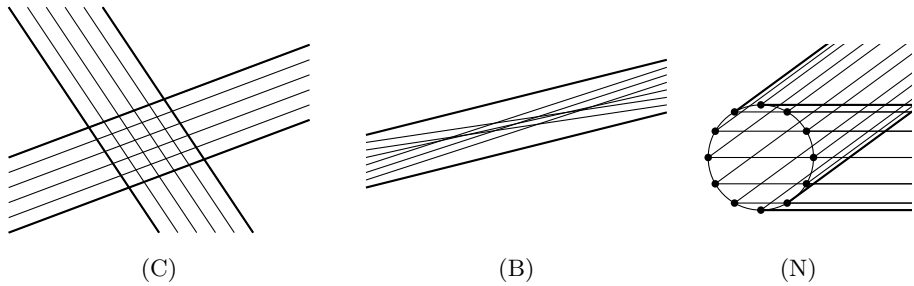
\begin{figure}[!ht]
       \centering
        \begin{tikzpicture}
                  \begin{scope}[xshift=-3cm]
        \foreach \i in {0,...,5} {
            \coordinate (N\i) at (\i/4,1);
            \coordinate (M\i) at (\i/4+2,-2);
            \coordinate (S\i) at (4,\i/5-0.5);
            \coordinate (P\i) at (0,\i/5-2);
             \draw (N\i) -- (M\i);
             \draw (S\i) -- (P\i); 
            }       
        \draw[thick] (N0) -- (M0);
        \draw[thick] (N5) -- (M5);
        \draw[thick] (P0) -- (S0);
        \draw[thick] (P5) -- (S5);
        \node at (2,-2.5){(C)};   
    \end{scope}
    \begin{scope}[xshift=1.75cm]
            \foreach \i in {1,...,8} {
            \coordinate (S\i) at (4,\i/10-0.5);
            \coordinate (P\i) at (0,\i/10-1.5);
            }      
        \draw[thick] (P1) -- (S1);
        \draw[thick] (P8) -- (S8);
        \draw (P2) -- (S5);
        \draw (P3) -- (S6);
        \draw (P4) -- (S7);
        \draw (P5) -- (S4);
        \draw (P6) -- (S2);
        \draw (P7) -- (S3);
        \node at (2,-2.5){(B)};   
    \end{scope}
    
        \begin{scope}[yshift=-1cm, xshift=7cm]
        \clip (-1,-1) rectangle (2,1.5);
        \draw (0,0) circle(0.7cm);
         \foreach \i in {1,...,12} {
            \coordinate (N\i) at (\i*360/12:0.7);
             \fill[black] (N\i) circle (0.05 cm);

           \draw (N\i) -- ([xshift=3cm] N\i); 
           \draw (N\i) -- ([xshift=3cm,yshift=2.2cm] N\i); 
                     }
          \draw[thick] (N4) -- ([xshift=3cm,yshift=2.2cm] N4);
          \draw[thick] (N10) -- ([xshift=3cm,yshift=2.2cm] N10);
          \draw[thick] (N3) -- ([xshift=3cm] N3); 
          \draw[thick] (N9) -- ([xshift=3cm] N9); 
        \end{scope}
                 \node at (7.5,-2.5){(N)}; 
        \end{tikzpicture}
        \caption{Possible crossings in the blow up: Bundle-bundle crossings (C), bundle crossings (B) and node crossings (N). }
        \label{fig:crossings}
    \end{figure}

Let us define these types (B), (C), and (N) more precisely and state their count. The corresponding counting process is described in the appendix. 

\begin{enumerate}
    \item[(C)] Two edge-bundles cross in a small neighborhood of a previous crossing in $D_{4}$. We call these \emph{bundle-bundle crossings} (C). Since each edge-bundle consists of $n^2$ edges, this gives $n^4$ bundle-bundle crossings for each  crossing in $D_4$. 

    \item[(B)] Two edges cross within a bundle. We call these \emph{bundle crossings} (B). Here we have $\binom{n}{2}^2$ crossings per bundle assuming $r(n) \ll n^{-1}$ and a suitable rotation of the circles.

    \item[(N)] Two edge-bundles cross at a node. We call these \emph{node crossings} (N).  Let $\alpha \in (0,\pi)$ be the angle between two incident edges $e,f$ in $D_4$ which were blown up to the edge-bundles, and let $\cro_\alpha$ be the resulting number of node crossings between the edges in the corresponding edge-bundles. Then we have: $\cro_\alpha + \cro_{\pi-\alpha} = \frac{n^3(n-1)}{2}$. 
\end{enumerate}

\subsection{Triangle densities in $D_{4}^{(n)}$}\label{subsec:trkkn}

The crossings in a triangle need to stem from bundle-bundle crossings (C), bundle crossings (B) or node crossings (N) as specified above. We first prove the following lemma.
\begin{lemma}
Let $D_4$ be a spherical drawing of a $K_{4,4}$ where each part consists of two pairs of antipodal vertices. Then no edge in $D_4$ is crossed twice. 
\end{lemma}

\begin{proof}
Let the parts of the $K_{4,4}$ be $A=\{v_1, \overline{v_1}, v_2, \overline{v_2}\}$ and $B=\{w_1,\overline{w_1},w_2, \overline{w_2}\}$. Note that the edge $v_1w_1$ can only be crossed by an edge between the other antipodal pairs, i.e. $v_2w_2,v_2\overline{w_2},\overline{v_2}w_2,\overline{v_2}\overline{w_2}$. All of them lie on the great circle defined by $v_2w_2$ so in fact only one of these edges can cross $v_1w_1$. By symmetry the same holds for the other edges. \qed
\end{proof}

We classify the triangles in the intersection graph of the blowup drawing $D_{4}^{(n)}$ as follows. We assign each crossing (which is an edge in the intersection graph) a \emph{type} (C), (B), or (N) depending on whether it is a bundle-bundle, within bundle or a node crossing. We say a triangle $c_1c_2c_3$ is of \emph{type} $(l(c_1)l(c_2)l(c_3))$ where $l(c_i)$ is the type of crossing
$c_i$.

 \begin{figure}[!ht]
       \centering
        \includegraphics[width=0.5\textwidth]{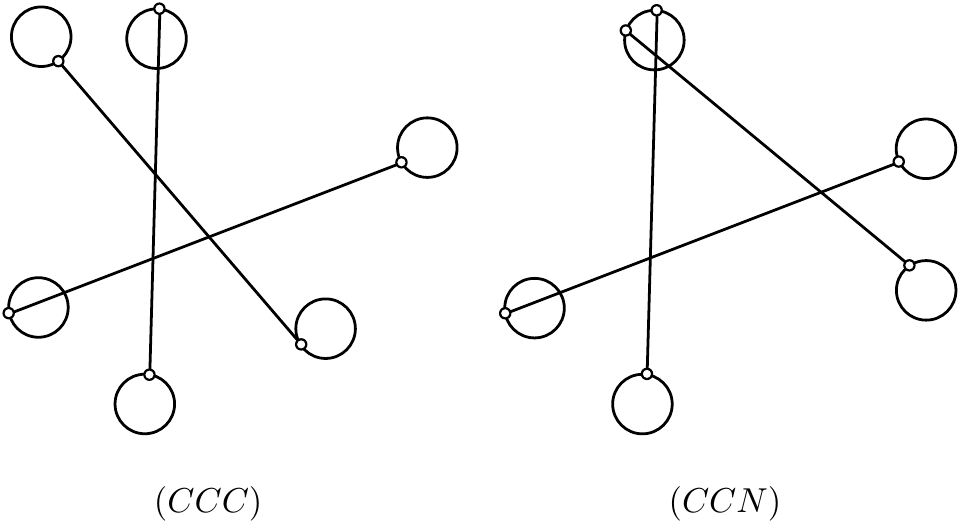}


        \caption{Triangles of type (CCC) and (CCN). }
        \label{fig:typescr}
    \end{figure}

The above lemma shows that there are no (CCC) or (CCN) triangles in $D_{4}^{(n)}$.
Also note that (CBB), (BBN) and (CBN) are not possible in general since BB suggests that all edges are from the same bundle and the bundled edges in (CBN) cross the third edge either at a node or at a bundle-bundle crossing but not at both. Triangles of type (NNN) either appear at three different nodes or at one node. However, we can not have (NNN) triangles at three different nodes since $K_{4,4}$ is bipartite and hence triangle-free. By the following lemma, the number of (NNN) triangles with all three crossings at one node is only of order $rn^6$ and can therefore be neglected.

\begin{figure}[!ht]
  \centering
  \includegraphics[width=0.8\textwidth]{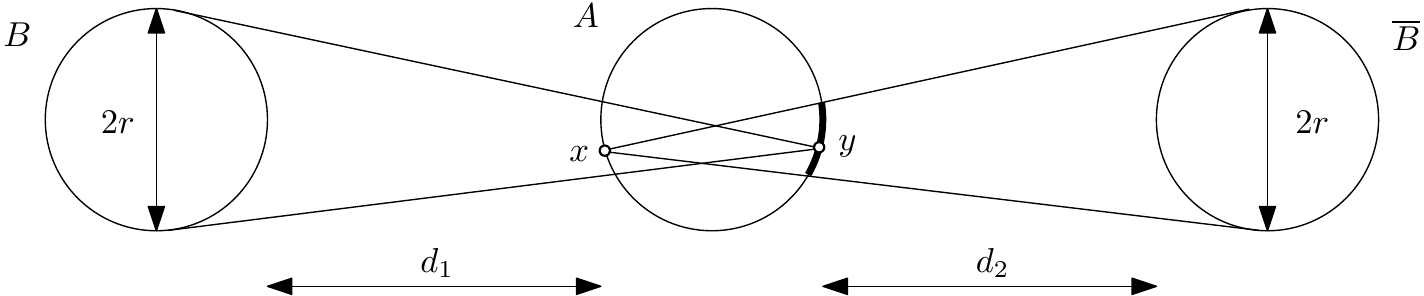}
  \caption{Two edges from node $A$ leading to antipodal nodes $B$ and $\overline{B}$ can cross. If $d=\min\{d_1,d_2\}$ and $r\le d$, then $|L_x| = O(rn)$.}
  \label{fig:NNNonenode}
\end{figure}

\begin{lemma}
\label{lem:no NNN at one node}
The number of (NNN) triangles in $D_4^{(n)}$ that correspond to three edges at the same node is $O(rn^6)$. Moreover, if $r(n) \ll n^{-1}$, there are no such triangles.
\end{lemma}

\begin{proof}
Let us refer to Figure \ref{fig:NNNonenode} and consider the possibility that an edge incident with a vertex $y$ and leading to a node $B$ crosses an edge incident with a vertex $x$ that leads to the antipodal node $\overline{B}$. If the geodesics from $x$ to $\overline{B}$ intersect the circle $C_A$ corresponding to $A$, we denote by $L_x$ the set of vertices in $A$ that are on the smallest circular arc that contains those intersections.

Then it is easy to see that either $x\in L_y$ or $y\in L_x$ (or both as shown in the figure). It can be shown (details can be found in the full paper) that the number of cases where $y\in L_x$ or $x \in L_y$ is $O(rn)$. In particular, if $r \ll n^{-1}$ then $L_x$ is empty. For each such pair $x,y$, the number of vertices $z$ whose incident edges leading to a node different from $B$ and $\overline{B}$ make an (NNN) crossing triangle with two edges incident with $x$ and $y$, respectively, is
$O((t+r)n)$, where $t$ is the number of vertices on the arc between $x$ and $y$. We define the parameter $l$ which is the number of vertices in the node $A$ between $x$ and the lowest point on the circle of $A$ (assuming that $x$ is in the lower half of the circle and on the left side). Then $t\in [2l - \Theta(rn), 2l + \Theta(rn)]$. This gives the following upper bound for the number of such triples $(x,y,z)$:
$$
   4 \sum_{l=1}^{n/4} O(rn) O(2l+rn) = O(r n^3).
$$
Finally, since each such triple involves $O(n^3)$ triples of mutually crossing edges incident with $x,y,z$, we confirm that the number of considered (NNN) triangles is $O(rn^6)$.\qed
\end{proof}

We are left with the following four cases.

(CNN) We consider pairwise crossings of three edges such that two cross at a bundle-bundle crossing and the third edge crosses one edge each at one node each. These crossings depend on the angles $\alpha,\beta,\gamma,\delta$ as depicted in Figure \ref{fig:sphere}. By Section (N) in Appendix B the number of pairs of vertices $x,y$ such that all edges at angle $\alpha$ incident to $x$ cross all horizontal edges incident to $y$ is $\frac{\pi- \alpha}{2\pi}n^2+O(rn^2+n)$. It is easy to see that the number of crossings we get in the triangle including $\alpha$ and $\beta$ is $\left( \frac{\pi- \alpha}{2\pi}\right)\left(\frac{\pi- \beta}{2\pi}\right) n^6 +O(rn^6+n^5)$. We have a similar count for the angles $\gamma$ and $\delta$. Then we have to add three other contributions corresponding to other crossings in $D_4$. The antipodal crossing involves a triangles with $\alpha,\beta$ and $\gamma,\delta$, whereas the other two crossings involve triangles with $\alpha, \gamma$ and $\beta, \delta$. Overall, this gives $\frac{2}{n^2}(\cro_{\alpha}+\cro_{\delta})(\cro_{\gamma}+\cro_{\beta})+O(rn^6+n^5)$ triangles of this kind. \\
(BBB) We consider pairwise crossings of three edges such that all edges are from one bundle. For each bundle we get ${n\choose 3}^2+O(rn^6)$ such triangles by Section (B) in Appendix B. There are $16$ bundles so in total we have $\sim \frac{4}{9} n^6+O(rn^6)$ triangles of the type (BBB).

(CCB) We consider pairwise crossings of three edges such that two edges are in one bundle and cross the third edge at a bundle-bundle crossing. There are $2{n\choose 2}^2 n^2+O(rn^6)$ triangles per each crossing in $D_4$. We have $4$ crossings so in total $\sim 2n^6+O(rn^6)$ triangles of this kind.  

(BNN) We consider pairwise crossings of three edges such that two are in the same bundle and cross the third edge at a node.  The argument is analogous to the one for crossings of type $(N)$.  Starting at the top vertex, we enumerate the vertices clockwise along the cycle as in Figure~\ref{fig:small_alpha}.  We consider an edge at angle $\alpha$ which ends in the $i$-th vertex in part $(A)$ and its crossings to horizontal edges. From Section (N) in the Appendix $B$, we know that  $\vert S_i \vert =2i+O(rn)$, where $S_i$ is as defined there. We can choose from ${2i+O(rn) \choose 2}$ pairs of left endpoints and ${n \choose 2}$ pairs of right endpoints for a triangle. The number of triangles with an edge ending in $i$ and another edge ending in a vertex in $W_x= \{ y \in A \mid x \in L_y \}$ is of order $O(rn^4)$, where $L_y$ is defined as in the proof of Lemma~\ref{lem:no NNN at one node}. We consider now edges at angle $\alpha$ ending in a vertex $x$ in $(B)$. Note that $ \vert S_x \vert = \frac{\pi-\alpha}{\pi} n +O(rn)$. We can choose for any one of ${\frac{\pi-\alpha}{\pi} n +O(rn) \choose 2}$ pairs of left endpoints ${n \choose 2}$ pairs of right endpoints for a triangle. The number of triangles with another edge ending in a vertex in $W_x$ is of order $O(rn^4)$. The contribution of triangles from edges in $(C)$ is the same as for edges in $(A)$. Hence the number of triangles of type (BNN) is 
        \begin{align*}
        2 \left(n \sum_{i=1}^{(\pi -\alpha)n/2\pi}{2i \choose 2} \cdot {n \choose 2}\right)+ \left(\frac{\alpha}{2\pi} n^2\right) \cdot  {\frac{\pi-\alpha}{\pi} n \choose 2} {n \choose 2} +O(rn^6+n^5).
    \end{align*}

For $\alpha$ and $\pi-\alpha$ added together, this gives
\begin{align*}
  \frac{1}{12} \, n^6 - \frac{\alpha(\pi-\alpha)}{8 \pi^2} \, n^6 +O(rn^6+n^5).
\end{align*}
Now note that for two bundles at angle $\alpha$ we can choose one of the bundles to contain the bundled edges. This gives two options. At each node we have two pairs of bundles meeting at angle $\alpha$ and two pairs of bundles meeting at angle $\pi-\alpha$. (In addition to these possibilities we get further (BNN) triangles from two bundles at the same node that lead to antipodal nodes and correspond to the value of $\alpha=\pi$. They give only $O(rn^6+n^5)$ triangles.) If $\alpha, \beta,\gamma,\delta$ are the angles as in Figure \ref{fig:sphere}, the overall number of (BNN) triangles is
\begin{align*}
     \frac{\alpha(\alpha-\pi)+\beta(\beta-\pi)+\gamma(\gamma-\pi)+\delta(\delta-\pi)}{ \pi^2} \, n^6 +\frac{8}{3}n^6+O(rn^6+n^5).
\end{align*}

If we leave out smaller order terms, the total number of triangles in the intersection graph by summing up the number of (CNN), (BBB), (CCB) and (BNN) triangles is
\begin{align*}
 \frac{\alpha^2+\beta^2+\gamma^2+\delta^2-\pi(\alpha+\beta+\gamma+\delta)}{ \pi^2} n^6 +  \frac{(2\pi-\alpha-\delta)(2\pi-\gamma-\beta)}{2 \pi^2} n^6+\frac{46}{9} n^6.
   \end{align*}

\begin{theorem}
\label{thm:K3 blowup}
Given a drawing $D_4$ of a $K_{4,4}$ where each part has two antipodal pairs, let $D_{4}^{(n)}$ be the blowup drawing, and let $\alpha,\beta,\gamma,\delta$ be the angles defined above. Then the limiting triangle density $t(K_3)$ of the sequence $D_4^{(1)}, D_4^{(2)}, \dots$ is equal to
\begin{align*}
\frac{3}{2^{12}\pi^2}\Bigl((2\pi-\alpha-\delta)(2\pi-\gamma-\beta) +2 (\alpha^2+\beta^2+\gamma^2+\delta^2)-2\pi (\alpha+\beta+\gamma+\delta)\Bigr) \\
+ \frac{23}{3\cdot 2^{10}}+O(r).
\end{align*}
\end{theorem}

\begin{proof}
We have determined the number of triangles in the intersection graphs.
Dividing by the number of possible triangles in the intersection graph, ${16n^2 \choose 3} = \frac{16^3}{6} n^6+O(n^5)$, gives the triangle density. \qed
\end{proof}

\section{Blowups as graphons}

Finally, let us show that the crossing graphs of drawings $D_{4}^{(n)}$ can be interpreted as certain graphons. 

\begin{theorem}
For fixed $r>0$ let $\mu_1$ and $\mu_2$ be uniform distributions over two pairs of antipodal circles on $\SS^2$ of radius $r$ each and let $W(\mu_1,\mu_2)$ be the crossing graph limit of corresponding drawings. If we consider blow-up drawings $D_{4}^{(n)}$  w.r.t. the centers of the circles of radius $r$, then the crossing graphs of $D_4^{(n)}$ converge and their limit is the graphon $W(\mu_1,\mu_2)$.
\label{thm:graphons}
\end{theorem}

\begin{proof}
All we need to show is that the density $t_1(H)$ in the random case limit and the density $t_2(H)$ of the blow-up drawing limit are the same for each graph $H$. Let $k = |H|$ be the number of vertices of $H$ and let $\phi: V(H)\to [k]$ be a bijection. For distinct points $x_1,\dots,x_k,y_1,\dots,y_k$ in $\SS^2$, let $X(x_1,\dots,x_k,y_1,\dots,y_k)$ be the intersection graph of the geodesic segments $x_1y_1,\dots,x_ky_k$.  Consider the following function
\begin{align*}
  f(x_1,\dots,x_k,y_1,\dots,y_k) = 
  \begin{cases}
   1, & \mbox{if $v\mapsto x_{\phi(v)}y_{\phi(v)}$ is a hom. $H\to X(x_1,\dots,y_k)$} \\
   0, & \mbox{otherwise.}
  \end{cases}
\end{align*}
Let $S_1$ and $S_2$ be the two circles on which $\mu_1$ and $\mu_2$ are defined, respectively. Since $f$ as defined above is measurable because $f^{-1}(1)$ is open, we can represent $t_1(H)$ as
\begin{align*}
    t_1(H)=\frac{1}{(8\pi r)^{k}}\int_{x \in S_1^n \times S_2^n} f(x) \, dx.
\end{align*}
In order to approximate $t_1(H)$ consider a set $C_n$ which consists of $n$ equidistant points on each of the cycles from $S_1,S_2$. Let $\pi_n: S_1\cup S_2 \to C_n$ be the function that maps a points from $S_1\cup S_2$ to its closest point in $X$. Let $g_n$ be a function  $g_n: (S_1\cup S_2)^{2n} \to (C_n)^{2n}$ that applies $\pi_n$ componentwise. Then $f_n=f \circ g_n$ converges pointwise to $f$ on $S_1^n \times S_2^n$. By the bounded convergence theorem  
\begin{align*}
    t_1(H)=\frac{1}{(8\pi r)^{k}} \int_{x \in S_1^n \times S_2^n} f(x) \, dx = \frac{1}{(8\pi r)^k} \lim_{n \to \infty} \int_{x \in S_1^n \times S_2^n} f_{n}(x) \, dx = t_2(H).\qed
\end{align*}
\end{proof}

The theorem shows that the same values for triangle densities in the $(\mu_1,\mu_2)$-random setting hold as for the blow-up limit in Theorem \ref{thm:K3 blowup}. 


\begin{theorem}
\label{thm:trianglebounds}
For fixed $r>0$ let $\mu_1$ and $\mu_2$ be uniform distributions over two pairs of antipodal circles on $\SS^2$ of radius $r$ each and let $W(\mu_1,\mu_2)$ be the crossing graph limit of the corresponding drawings. Then 
\begin{align*}
 \frac{83}{3\cdot 2^{12}} + O(r) \leq t(K_3, W(\mu_1,\mu_2)) \leq \frac{1}{3\cdot 2^{5}} + O(r),
\end{align*}
and these bounds are best possible. The limiting triangle density $t(K_3)$ depends on the angles $\alpha, \beta, \gamma, \delta$, and any value in the interval $\bigl(\frac{83}{12288}, \frac{128}{12288}\bigr)$ is possible.
\end{theorem}

The proof is in the appendix.

\section{Conclusion}

It should be noted that the proofs of Theorem \ref{thm:convergence} and Theorem \ref{thm:K2 density} also extend to the case of the complete graph $K_n$ where we choose $n$ points from the sphere with respect to some antipodally symmetric probability measure $\mu$. 
(Let us observe that antipodal symmetry is needed for such a result.)
In Theorem \ref{thm:K3 blowup} the value $\frac{23}{3\cdot 2^{10}}=0.00748$ appears which is included in the interval given by Theorem \ref{thm:trianglebounds}. Numerical experiments show that the triangle density with respect to the uniform distribution is close to $0.0075$. This matches the mentioned special value from the blow-up setting. It would be of interest to study the crossing graph limit for drawings on the sphere of the complete graph or the complete bipartite graph when we restrict our probability measure to a uniform measure on the sphere. As Moon already showed in 1965 \cite{Moon65}, it holds asymptotically almost surely that $t(K_2)=\frac{1}{8}$, so it would be of interest to find a closed expression for $t(K_3)$. 

\bibliographystyle{splncs04}
\bibliography{biblio_geodesicS2}

\appendix

\section{Proof of Proposition \ref{prop:varY_H}}

\begin{proof}[of Proposition \ref{prop:varY_H}]
By definition
\begin{align}
var(Y_H)&=\EE[(Y_H-\EE[Y_H])^2] \nonumber \\
&=E\left[\left(\sum_{ \phi: V(H)\to V(X_n)} y_{H,\phi}-E_\phi\right)^2\right] \nonumber \\
&=\sum_{ \phi: V(H)\to V(X_n)} \sum_{ \phi': V(H)\to V(X_n)} \EE[(y_{H,\phi}-E_\phi) (y_{H,\phi'}-E_{\phi'})]. \label{eq:varY_H}
\end{align}
For independent variables $y_{H,\phi}$ and $y_{H,\phi'}$ the expectation $ \EE[(y_{H,\phi}-E_\phi) (y_{H,\phi'}-E_{\phi'})]$ equals zero so we only need to consider those pairs $\phi$ and $\phi'$ for which $y_{H,\phi}$ and $y_{H,\phi'}$ are dependent.

The events ``$\phi$ is a graph homomorphism $H\to X_n$" and ``$\phi'$ is a graph homomorphism $H\to X_n$" are independent if $im(\phi)=\{e_1,\dots,e_k\}=\{v_1w_1,\dots,v_kw_k\}$ and $im(\phi')=\{e'_1,\dots,e'_k\}=\{v'_1w'_1,\dots,v'_kw'_k\}$ satisfy
\begin{align*}
\left\vert\{v_1,\dots,v_{k},w_1,\dots,w_{k}\} \cap \{v'_1,\dots,v'_{k},w'_1,\dots,w'_{k}\}   \right\vert \leq 1.
\end{align*}
But note that for these sets to share at least two points we have $\binom{n}{2}$ choices for those two special points and at most $(n^{2k-2})^2$ for the remaining ones. The number of edges $(e_1,\dots,e_k)$ that can be formed by a set of vertices in $X_n$ $\{v_1,\dots,v_{k},w_1,\dots,w_{k}\}$ does not depend on $n$, so we have at most $O(n^{4k-2})$ pairs $y_{H,\phi}$ and $y_{H,\phi'}$ that are dependent as $\phi$ and $\phi'$ are defined by $(e_1,\dots,e_k)$ and $(e'_1,\dots,e'_k)$ only. Note that for each pair
\begin{align*}
   \left\vert \EE[(y_{H,\phi}-\mu_\phi) (y_{H,\phi'}-\mu_{\phi'})]\right\vert \leq 1
\end{align*}
since $\vert y_{H,\phi}(X_n)-\mu_\phi\vert\leq 1$ for any $X_n$. Summing up those expectations over the dependent variables, (\ref{eq:varY_H}) gives $var(Y_H)=O(n^{4k-2})$. \qed
\end{proof}

\section{Counting crossings of types (C), (B), and (N)}

To help us with counting crossings of type $(N)$ below, we first prove the following Lemma. 


\begin{lemma}
\label{lem: O(rn) bound}
Let $A$ and $B$ be two nodes corresponding to adjacent vertices in $D_4$ and let $x \in A$. If the geodesics from $x$ to $B$ intersect the circle $C_A$ corresponding to $A$, we denote by $L_x$ the set of vertices in $A$ that are on the smallest circular arc that contains those intersections.
Then $\vert L_x \vert = O(rn)$. Moreover, if $W_y = \{ x\in A \mid y \in L_x \}$, then $\vert W_y\vert = O(rn)$.
\end{lemma}

\begin{proof}
Note that the length of $L_x$ is $O(r^2/d)$, where $d$ is the distance from $A$ to $B$ (see Figure \ref{fig_O(rn)}). As we consider the distance $d$ to be constant, the number of vertices $y$ such that $y\in L_x$ is $O(r^2n/(2\pi r))=O(rn)$. Moreover, since the angles at $y$ and $x$, as shown in Figure \ref{fig_O(rn)}, are almost the same as $d$ is large compared to $r$, we also have $\vert W_y \vert=O(rn)$. \qed
\end{proof}

\begin{figure}[!ht]
       \centering
       \includegraphics[width=0.65\textwidth]{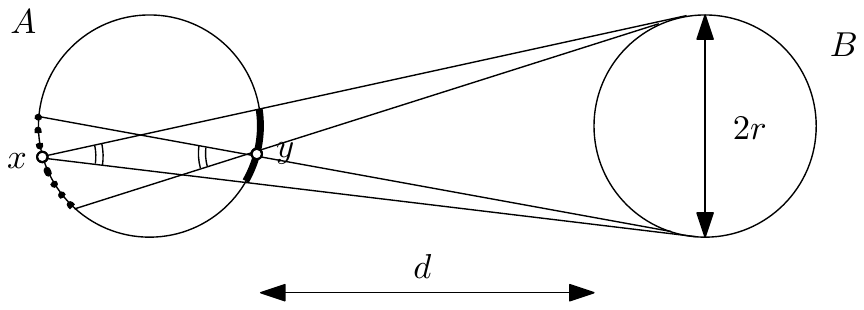}
       \caption{$L_x$ are the vertices on the arc between the extremal two edges leading from $x$ to $B$. The dashed arc in this figure contains vertices in $W_y$.  }
        \label{fig_O(rn)}
\end{figure}

In the following we discuss and count the crossings of each type.

\begin{enumerate}
    \item[(C)] Two edge-bundles cross in a small neighborhood of a previous crossing in  $D_{4}$. (We assume that $r(n)$ is small.) We call these \emph{bundle-bundle crossings} (C). Since each edge-bundle consists of $n^2$ edges, this gives $n^4$ bundle-bundle crossings for each  crossing in $D_4$. 
    \item[(B)] Two edges cross within a bundle. We call these \emph{bundle crossings} (B). Here we have $\binom{n}{2}^2$ crossings per bundle if $r(n) \ll n^{-1}$ and suitably rotated circles considering the following elementary argument:

\begin{claim}\label{cl:crKnm}
Let $D_{A,B}$ be the subdrawing of $D_4^{(n)}$ consisting of all edges between two nodes $A,B$ corresponding to two adjacent vertices of $D_4$. If $r(n) \ll n^{-1}$ and the cycles are suitably rotated, then $\cro(D)=\binom{n}{2}^2$.
\end{claim}

\begin{proof}
Any 4-tuple of two vertices from $A$ and two vertices from $B$ determines precisely one crossing, and each crossing corresponds to precisely one such 4-tuple of vertices.
\qed
\end{proof}

If we drop the restriction $r(n) \ll n^{-1}$ and consider general $r$ the picture looks slightly different. Referring to Figure $\ref{fig_O(rn)}$ we can see that if $y \in L_x$ then the pair $x,y$ does not contribute (B) crossings with any pair of vertices in $B$. For another pair of vertices $w,z$ in $B$ we can see that the edges from $x$ to $w,z$ and from $y$ to the antipodals $\overline{w},\overline{z}$ contribute two crossings. Hence generally we have $\binom{n}{2}^2+O(rn^4)$ bundle crossings and $O(rn^4)$ additional node crossings.

    \item[(N)] Two edge-bundles cross at a node. We call these \emph{node crossings} (N).  Let $\alpha \in (0,\pi)$ be the angle between two incident edges $e,f$ in $D_4$ which were blown up to the edge-bundles, and let $\cro_\alpha$ be the resulting number of node crossings between the edges in the corresponding edge-bundles. We consider one bundle to be horizontal whereas the other bundle is counterclockwise at angle $\alpha$. We partition the edges from the bundle at angle $\alpha$ into four sets depending on which vertex in the node they are adjacent to. Starting at the top vertex, we enumerate the vertices clockwise along the cycle. The first $\frac{\pi-\alpha}{2\pi}n$ vertices\footnote{The numbers of nodes in each part are rounded up or down, but these changes will make our counts of crossings deviate only in a lower order term and can thus be neglected.} belong to part $(A)$, the next $\frac{\alpha}{2\pi}n$ vertices belong to part $(B)$, then we have $\frac{\pi-\alpha}{2\pi}n$ vertices belonging to part $(C)$ and the last $\frac{\alpha}{2\pi}n$ vertices belong to part $(D)$ as in 
    Figure~\ref{fig:small_alpha}. Assuming the circle radius $r$ is small enough, all edges in one bundle are almost parallel to each other up to an error term that depends on $r$. For each vertex $x$ we introduce two sets of vertices, $S_x$ and $W_x$. Let $y\in S_x$ if all horizontal edges incident with $y$ cross all edges at angle $\alpha$ that are incident with $x$, and let $W_x=\{ y \in A \mid x \in L_y \}$ where $L_y$ is defined as in Lemma \ref{lem:no NNN at one node} with respect to the horizontal edges. If $x$ is the $i$-th vertex in $(A)$ then $\vert S_x \vert = 2i+O(rn)$ and $\vert W_x \vert = O(rn)$. If $x$ is in $(B)$ then $\vert S_x \vert = \frac{\pi-\alpha}{\pi} n+O(rn)$ and $|W_x| = O(rn)$. If $x$ is in $(C)$ then a similar count as in $(A)$ applies if we enumerate those vertices starting at the last vertex in $(C)$. If $x \in (D)$ then $S_x$ is empty and $\vert W_x \vert = O(rn)$. Note that each pair $x,y$ such that $y \in S_x$ contributes $n^2$ crossing and if $y\in W_x$ then the contribution is $O(n^2)$ crossing. Finally the number of crossings is
    \begin{align*}
        \cro_\alpha(r) &= 2 \left( \sum_{i=1}^{(\pi -\alpha)n/2\pi} (2i+O(rn)) \cdot n^2\right) + \left(\frac{\alpha}{2\pi} n^2\right) \, \left(\frac{\pi-\alpha}{\pi}+O(rn)\right)\, n^2 \\
        & = \frac{\pi-\alpha}{2 \pi}\cdot n^4 + O(rn^4+n^3).
        \end{align*}

\begin{figure}[!ht]
    \centering
\begin{tikzpicture}
            \draw (0,0) circle(2.2cm);
            \clip (-3,-3) rectangle (3,3);
            \coordinate (N1) at (0*360/4:2.2);
            \coordinate (M1) at (3*360/16:2.2);
            \coordinate (N2) at (1*360/4:2.2);
            \coordinate (M2) at (7*360/16:2.2);
            \coordinate (N3) at (2*360/4:2.2);
            \coordinate (M3) at (11*360/16:2.2);
            \coordinate (N4) at (3*360/4:2.2);
            \coordinate (M4) at (15*360/16:2.2);
            \foreach \i in {1,2,3,4}{
             \fill[black] (N\i) circle (0.05 cm);
             \fill[black] (M\i) circle (0.05 cm);
             }
            \coordinate (N11) at (0*360/4:3.2);
            \coordinate (M11) at (3*360/16:3.2);
            \coordinate (N12) at (1*360/4:3.2);
            \coordinate (M12) at (7*360/16:3.2);
            \coordinate (N13) at (2*360/4:3.2);
            \coordinate (M13) at (11*360/16:3.2);
            \coordinate (N14) at (3*360/4:3.2);
            \coordinate (M14) at (15*360/16:3.2); 
            \draw (N11) -- (N13);
            \draw (M11) -- (M13);
            \draw (N12) -- (N14);
            \draw (M12) -- (M14);
            \draw [red,line width=3pt,domain=-22.5:90] plot ({2.2*cos(\x)}, {2.2*sin(\x)});
            \draw [red,line width=3pt,domain=157.5:270] plot ({2.2*cos(\x)}, {2.2*sin(\x)});
            \draw [blue,ultra thick,dotted,domain=270:337.5] plot ({2.2*cos(\x)}, {2.2*sin(\x)});
            \node[red] at (2.2,1.6) {$(A)$};      
            \node[red] at (-2.2,-1.6) {$(C)$};  
            \node[blue] at (1.6,-2.2) {$(B)$}; 
            \node[black] at (-1.6,2.2) {$(D)$};
            \node at (1.5*360/16:0.6){$\alpha$};   
            \draw [thick,domain=0:67.5] plot ({cos(\x)}, {sin(\x)});
\end{tikzpicture}
    \caption{$(A)$, $(B)$, $(C)$ and $(D)$ are special areas of vertices of a node. An illustration where $0 < \alpha \leq \frac{\pi}2$.}
    \label{fig:small_alpha}
\end{figure}
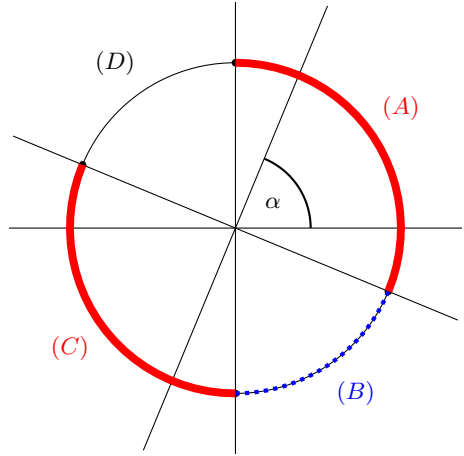

\begin{figure}[!ht]
       \centering
       \includegraphics[width=0.48\textwidth]{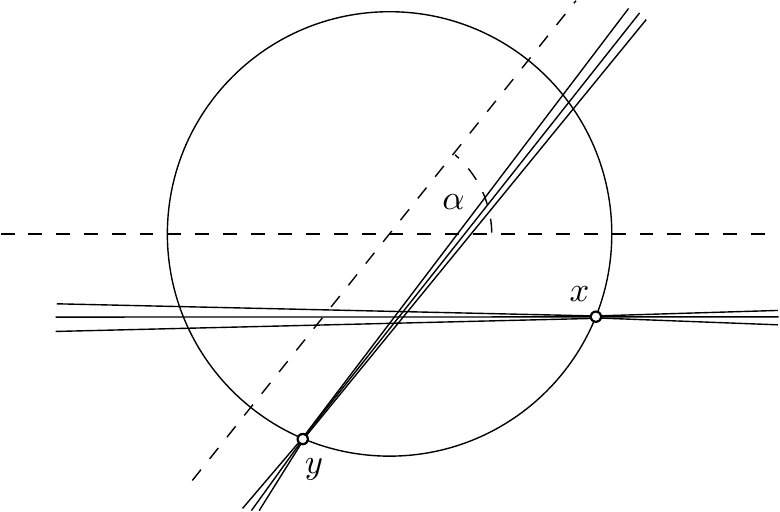}
       \caption{An illustration of the antipodal argument for the total number of node crossings at one of the nodes. The edges incident with any $x$ and $y$ ($x\ne y$) yield precisely $n^2$ node crossings.}
        \label{fig:antipodalargument}
\end{figure}
    
Let us observe that $\cro_\alpha + \cro_{\pi-\alpha} = \frac{n^3(n-1)}{2}$. This formula, which is exact, can be obtained directly by considering all four bundles arriving to a node as shown in Figure~\ref{fig:antipodalargument}. It is apparent from the figure that the total number of crossings, which is equal to $2\cro_\alpha + 2\cro_{\pi-\alpha}$, can be counted by considering any ordered pair $(x,y)$ of distinct vertices in the node and observing that the horizontal edges incident with $x$ and the angle $\alpha$ edges incident with $y$ leaving in both directions from each vertex yield $n^2$ crossings. 
\end{enumerate}

\section{Proof of Theorem \ref{thm:trianglebounds}}

\begin{lemma}
Let $\alpha, \beta, \gamma, \delta$ be as in Figure \ref{fig:sphere}. Then $\alpha+\beta+\gamma+\delta < 2\pi$. 
\label{lem:anglesum}
\end{lemma}
\begin{proof}
We refer to Figure $\ref{fig:sphere}$. Let $T$ be the triangle formed by $w_2$,$v_2$ and the crossing $c$ of the segments $w_2v_1$ and $v_2w_1$. Let $T'$ be the triangle formed by $\overline{v}_1, \overline{w}_1$ and the crossing $c$ and note that $T'$ contains $T$. Note that the angle $a$ at $\overline{v}_1$ and the angle $b$ at $\overline{w}_1$ in $T'$ are $a=\pi-\gamma$ and $b=\pi-\beta$. As $T$ is within $T'$ its area is smaller and hence its angular defect is smaller which is proportional to the angle sum. This tells us that $\alpha+\beta < a+b=2\pi-\gamma-\delta$ which proves the lemma.\qed
\end{proof}

\begin{proof}[of Theorem \ref{thm:trianglebounds}]
By Theorem \ref{thm:graphons} we can refer to Theorem \ref{thm:K3 blowup} to find the extremal bounds. We give a proof for the upper bound first, which is attained if all angles are close to zero as in Figure \ref{fig:sphere2}. This is optimal since rewriting the equation from Theorem \ref{thm:K3 blowup} gives
\begin{align*}
\frac{3}{2^{12}\pi^2} & ( 4\pi^2- 4\pi(\alpha+\beta+\gamma+\delta) + (\alpha+\delta)(\gamma+\beta)+2\alpha^2+2\beta^2+2\gamma^2+2\delta^2) \\
& ~ + \frac{23}{3\cdot 2^{10}}+O(r) \\
\leq &~ \frac{3}{2^{12}\pi^2} (- 4\pi(\alpha+\beta+\gamma+\delta)+(2 \pi)(\gamma+\beta) + 2\pi \alpha+2\pi \beta+ 2\pi \gamma+2\pi \delta) \\
&~ + \frac{1}{3\cdot 2^{5}}+O(r) \\
\leq &~ \frac{1}{3\cdot 2^{5}}+O(r).
\end{align*}

    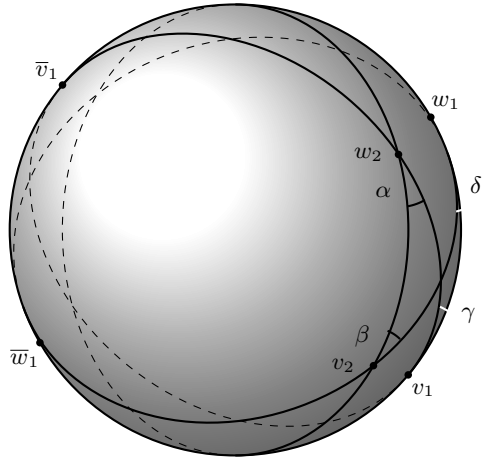
\begin{figure}
        \centering
\begin{tikzpicture}
   \pgfmathsetmacro\R{3} 
\shadedraw[ball color = white,thick] (0,0) circle (\R);

        \foreach \angle[count=\n from 1] in {20} {

          \begin{scope}[rotate=\angle]
            \path[draw,dashed,name path global=d\n] (\R,0) arc [start angle=0,
              end angle=180,
              x radius=\R cm,
            y radius=2.5cm] ;
            \path[draw,thick,name path global=s\n] (-\R,0) arc [start angle=180,
              end angle=360,
              x radius=\R cm,
            y radius=2.5cm] ;
          \end{scope}
        }

  \foreach \angle[count=\n from 2] in {-40,270
  } {
          \begin{scope}[rotate=\angle]
            \path[draw,thick,name path global=s\n] (\R,0) arc [start angle=0,
              end angle=180,
              x radius=\R cm,
            y radius=2.3cm] ;
            \path[draw,dashed,name path global=d\n] (-\R,0) arc [start angle=180,
              end angle=360,
              x radius= \R cm,
            y radius=2.3cm] ;
          \end{scope}
          }

    \InterSec{s1}{s3}{I3} ;
    \InterSec{s1}{s2}{I2} ;
    \InterSec{s3}{s2}{I1} ;
    
    \InterSec{d1}{d2}{J3} ;
    \InterSec{d1}{d3}{J2} ;
    \InterSec{d3}{d2}{J1} ;    
    \fill[black] (I3) circle (0.05 cm) node[xshift=-0.4cm]{$v_2$};  
  \fill[black] (I1) circle (0.05 cm)node[xshift=-0.4cm]{$w_2$};  ;  
      \fill[black] (3*cos{140}, 3*sin{140}) circle (0.05 cm) node[xshift=-0.2cm, yshift=0.2cm]{$\overline{v}_1$};  
            \fill[black] (3*cos{320}, 3*sin{320}) circle (0.05 cm) node[xshift=0.2cm, yshift=-0.2cm]{$v_1$}; 
                  \fill[black] (3*cos{210}, 3*sin{210}) circle (0.05 cm) node[xshift=-0.2cm, yshift=-0.2cm]{$\overline{w}_1$};  
            \fill[black] (3*cos{30}, 3*sin{30}) circle (0.05 cm) node[xshift=0.2cm,yshift=0.2cm]{$w_1$}; 
            \draw[thick, shift=(I1)] plot [domain=280:300] ({0.7*cos(\x)}, {0.7*sin(\x)});
            \node[xshift=-0.2cm,yshift=-0.5cm] at (I1) {$\alpha$};
            
            \draw[thick, shift=(I3)]  plot [domain=45:68] ({0.5*cos(\x)}, {0.5*sin(\x)});
            \node[xshift=-0.15cm,yshift=0.4cm] at (I3) {$\beta$};
            \draw[thick,white] plot [domain=50:65] ({1*cos(\x)+3*cos(320)}, {1*sin(\x)+3*sin(320)});
            \node[xshift=0.8cm,yshift=0.8cm] at (3*cos{320}, 3*sin{320}) {$\gamma$};
            \draw[thick,white]  plot [domain=285:360] ({1.3*cos(\x)+3*cos(30)}, {1.3*sin(\x)+3*sin(30)});
            \node[xshift=0.6cm,yshift=-0.9cm] at (3*cos{30}, 3*sin{30}) {$\delta$};         
\end{tikzpicture}
        \caption{If $w_2,v_2$ approach $w_1,v_1$, respectively, then all angles $\alpha,\beta,\gamma,\delta$ converge to zero.}
        \label{fig:sphere2}
    \end{figure}

The claimed value in the lower bound in Theorem \ref{thm:K3 blowup} is attained for $\alpha=\beta=\gamma=\delta=\frac{\pi}{2}$. To construct an example where all values are close to $\frac{\pi}{2}$, we exchange $v_1$ and $\overline{v}_1$ with $w_1$ and $\overline{w}_1$ in  Figure \ref{fig:sphere2}, respectively. To show that we can not do better let $\alpha=\frac{\pi}{2}+a$, $\beta=\frac{\pi}{2}+b$, $\gamma=\frac{\pi}{2}+c$ and $\delta=\frac{\pi}{2}+d$. Omitting O(r) terms, the associated triangle density is
\begin{align*}
\frac{3}{2^{12}\pi^2} &\left((\pi-a-d)(\pi-c-b) +2 a^2+2b^2+2c^2+2d^2-2\pi^2 \right)+
\frac{23}{3\cdot 2^{10}}\\
= &~ \frac{83}{3\cdot 2^{12}}+ \pi (-a-b-c-d) \\
&~ + \frac{1}{2}(a+b+c+d)^2+\frac{1}{2}(a-d)^2+\frac{1}{2}(b-c)^2+a^2+b^2+c^2+d^2
\end{align*}
and it attains its global minimum at $a=b=c=d=0$ as $-a-b-c-d \geq 0$ by Lemma $\ref{lem:anglesum}$.

As we can come from any arrangement of four antipodal pairs of points to any other arrangement by continuously changing the points, and since $r$ can be arbitrary small, any triangular density in the interval $\bigl(\frac{83}{12288}, \frac{128}{12288}\bigr)$
can be attained with one of these graphons. \qed
\end{proof}

\section{Sketch of the proof of Theorem \ref{thm:Mo}}

As ref.~\cite{Mo20} is not available in its final form at this time, we add a sketch of the proof of Theorem \ref{thm:Mo}. It is based on the following result, see \cite{Mo19}.

\begin{theorem}
\label{thm:Z}
Suppose that $n>0$ is an even integer and that $P,Q$ are disjoint sets, each containing $n/2$ points, on the unit sphere in general position. Let $D$ be the geodesic drawing of $K_{n,n}$, where points in $P\cup \overline P$ and $Q\cup \overline Q$ are the vertices of the bipartition of $K_{n,n}$. Then $cr(D) = Z(n,n)$. 
\end{theorem}

\begin{proof}[of Theorem \ref{thm:Mo}, sketch]

By the nondegeneracy we can assume that $D_n$ has no antipodal vertices and let $P,Q$ be the vertices of the bipartition of the $K_{n,n}$. Let $D$ be the corresponding drawing of $K_{2n,2n}$ with parts $P \cup \overline{P}$ and $Q \cup \overline{Q}$. We can partition the set of antipodal geodesic drawings of $K_{2n,2n}$ on the sphere into classes of equivalent drawings, where two drawings are isomorphic if there exists a homeomorphism of the sphere which transforms one into the other. There are only finitely many equivalence classes of geodesic drawings of the complete bipartite graph with $2n$ vertices in each part. Let $C_1,\dots, C_m$ be those equivalence classes. Considering drawings $D$ in $C_i$, if we delete one vertex from each antipodal pair uniformly at random, we get a drawing $D_n$. The drawing $D$ has $Z(n,n)$ crossings by Theorem \ref{thm:Z} and one of those crossings appears in $D_n$ if and only if all the involved vertices are in $D_n$, which is with probability $\frac{1}{16}$. By linearity of expectation
\begin{align*}
 \EE(cr(D_n) \, \vert \, D \in C_i) = \frac{1}{16} Z(2n,2n)= \frac{1}{16}n^2(n-1)^2.
\end{align*}
Using the law of total expectation we get
\begin{align*}
    \EE(cr(D_n)) = \sum_{i=1}^m P(D \in C_i) \cdot \EE(cr(D_n) \, \vert \, D \in C_i) = \frac{1}{16} n^2(n-1)^2.
\end{align*}
Since $Z(n,n) \sim \frac{1}{16}n^4$ the theorem follows. 
\qed
\end{proof}

\end{document}